\documentclass{llncs}

\usepackage[utf8]{inputenc}
\let\defaultemptyset=\emptyset
\usepackage{
enumitem,
mathpartir,
amssymb,%
stmaryrd,cmll,
graphicx,empheq,wrapfig,
picinpar,
yhmath
} %
\renewcommand{\emptyset}{\defaultemptyset}
\usepackage[mathcal]{euscript}
\usepackage{fullshort}
\setboolean{fullpaper}{true}
\usepackage{url}

\newlist{myaxioms}{enumerate}{10}
\setlist[myaxioms,1]{label=(P\arabic*)}
\setlist[myaxioms]{resume}

\newlist{enumeratei}{enumerate}{10}
\setlist[enumeratei]{label=\emph{(\roman*)}}

\DeclareMathOperator{\ob}{ob}

\DeclareMathOperator{\domv}{dom_{\mathit{v}}}
\DeclareMathOperator{\codv}{cod_{\mathit{v}}}
\DeclareMathOperator{\domh}{dom_{\mathit{h}}}
\DeclareMathOperator{\codh}{cod_{\mathit{h}}}

\DeclareMathOperator{\im}{Im}

\newcommand{\restr}[2]{#1_{|#2}}

\newcommand{\two}{\mathsf{2}}

\newcommand{\bigcat}[1]{\ensuremath{\mathsf{#1}}}
\newcommand{\maji}[1]{\ensuremath{\mathbb{#1}}}

\newcommand{\commentthis}[1]{}

\newcommand{\Psh}[1]{\widehat{#1}}

\renewcommand{\P}{\maji{P}}

\newcommand{\Cospan}[1]{\bigcat{Cospan} (#1)}

\newcommand{\into}{\hookrightarrow}
\newcommand{\otni}{\hookleftarrow}

\newcommand{\ot}{\leftarrow}
\newcommand{\xto}[1]{\xrightarrow{#1}}
\newcommand{\xot}[1]{\xleftarrow{#1}}
\newcommand{\xTo}[1]{\xRightarrow{#1}}

\newcommand{\xinto}[1]{\xhookrightarrow{#1}}

\newcommand{\labelof}[1]{\ell^{#1}}

\newcommand{\labelpi}{\labelof{\picalc}}

\newcommand{\faireqof}[2]{\mathrel{\sim^{#1}_{#2}}}

\renewcommand{\hat}[1]{\widehat{#1}}

\def\framed{%
\setbox0=\vbox\bgroup%
\advance\hsize by -2\fboxsep\advance\hsize by -2\fboxrule%
\linewidth=\hsize%
}
\def\endframed{%
\egroup\noindent\framebox[\textwidth]{\box0}\vspace*{1mm}}

\usepackage{tikz,tikz-cats}
\tikzset{todim/.style = {decoration={markings, mark=at position .5 with %
      {\draw (-1pt,-1pt) rectangle (1pt,1pt);}},postaction={decorate}}}
\renewcommand{\diaggrandhauteur}{.6}
\renewcommand{\diaggrandlargeur}{1.3}

\usepackage{color}
\definecolor{dkgreen}{rgb}{0,0.2,0}

\newcommand{\A}{\maji{A}}

\newcommand{\B}{\maji{B}}
\newcommand{\Beh}[1]{\Behfun_{#1}}
\newcommand{\Behfun}{\bigcat{B}}

\newcommand{\C}{\maji{C}}

\newcommand{\tick}{\daimon}

\newcommand{\Chat}{\hat{\maji{C}}}
\newcommand{\Chatf}{\Pshf{\maji{C}}}

\newcommand{\D}{\maji{D}}

\renewcommand{\DH}{\D_{H}}

\newcommand{\Dh}{\D_h}

\newcommand{\Dv}{\D_v}

\newcommand{\ccs}{\mathit{CCS}}
\newcommand{\picalc}{\mathit{Pi}}
\newcommand{\Dccs}{\maji{D}^{\scriptscriptstyle \ccs}}

\newcommand{\E}{\maji{E}}

\newcommand{\HH}{\bigcat{H}}

\newcommand{\LL}{\bigcat{L}}
\newcommand{\LLL}{\mathcal{L}}

\newcommand{\MMMB}{\griso{\B}}

\newcommand{\V}{\maji{V}}

\newcommand{\W}{\maji{W}}

\newcommand{\yoneda}{\bigcat{y}}

\newcommand{\RR}{\bigcat{R}}

\renewcommand{\SS}{\bigcat{S}}
\newcommand{\SSS}{\mathcal{S}}

\newcommand{\TT}{\bigcat{T}} %
\renewcommand{\tt}{\bigcat{t}} %

\newcommand{\Set}{\bigcat{Set}}
\newcommand{\set}{\bigcat{set}}
\newcommand{\ford}{\bigcat{ford}}

\newcommand{\OPsh}[1]{\wideparen{#1}}
\newcommand{\FPsh}[1]{\overline{#1}}
\newcommand{\Pshf}[1]{\widehat{#1}^{{}_f}}

\newcommand{\Gph}{\bigcat{Gph}}

\newcommand{\Nat}{\mathbb{N}} 

\newcommand{\un}{I}
\newcommand{\para}{\mathbin{\mid}}

\newcommand{\transl}[1]{\llbracket #1 \rrbracket}
\newcommand{\translfun}{\llbracket - \rrbracket}

\newcommand{\rond}{\circ}
\newcommand{\vrond}{\mathbin{\bullet}}

\newcommand{\id}{\mathit{id}}
\newcommand{\iso}{\cong}

\newcommand{\impll}{\multimap}
\newcommand{\tens}{\otimes}

\newcommand{\ens}[1]{\{ #1 \}}

\newcommand{\aalt}{\mathrel{|}}

\newcommand{\op}[1]{#1^{\mathit{op}}}

\newcommand{\send}[2]{\ensuremath{\bar{#1}\langle#2\rangle}}
\newcommand{\sendab}{\send{a}{b}}

\newcommand{\Gam}{\Gamma}
\newcommand{\Del}{\Delta}

\newcommand{\labelD}{\ell_\D}

\newcommand{\bureaucratic}[1]{}

\newcommand{\paraof}[1]{\pi_{#1}}

\newcommand{\paralof}[1]{\pi^l_{#1}}
\newcommand{\pararof}[1]{\pi^r_{#1}}

\newcommand{\paralp}{\paralof{p}}
\newcommand{\pararp}{\pararof{p}}

\newcommand{\paraofij}[2]{\mathbin{{}_{#1}\!\para\!{}_{#2}}}

\newcommand{\nuof}[1]{\nu_{#1}}
\newcommand{\nun}{\nuof{n}}
\newcommand{\nup}{\nuof{p}}

\newcommand{\tickof}[1]{\tick_{#1}}
\newcommand{\tickn}{\tickof{n}}
\newcommand{\tickp}{\tickof{p}}

\newcommand{\taunamcd}{\tau_{n,a,m,c,d}}

\newcommand{\lts}{\textsc{lts}}
\newcommand{\anlts}{an \lts{}}
\newcommand{\ltss}{\lts{s}}
\newcommand{\Lts}{\textsc{Lts}}
\newcommand{\Ltss}{\Lts{s}}

\newcommand{\ortho}{\mathrel{\bot}}

\newcommand{\botright}[1]{#1^{\bot}}
\newcommand{\botleft}[1]{{{}^{\bot}#1}}
\newcommand{\botrightleft}[1]{\botleft{(\botright{#1})}}

\newcommand{\daimon}{\heartsuit}

\newcommand{\state}{\sigma}

\newcommand{\deriv}{\partial}

\renewcommand{\with}[1]{\langle #1 \rangle}

\newcommand{\extafun}[1]{\overline{(-)}}

\newcommand{\Sierp}{\Sigma}

\tikzset{history/.style = {-open triangle 45}}

\newcommand{\vdashdefinite}{\vdash_{\mathsf{D}}}

\newcommand{\outnab}{o_{n,a,b}}

\newcommand{\outmcd}{o_{m,c,d}}

\newcommand{\inna}{\iota_{n,a}}

\newcommand{\paradr}{\shortdfibto}

\usepackage{preamble}
\renewcommand{\MMMB}{\mathbb{B}}

\begin{document}
\title{Fully-abstract concurrent games for $\pi$}
\titlerunning{Concurrent games for $\pi$} %
\author{Clovis Eberhart\inst{1} 
\and Tom Hirschowitz\inst{2}\thanks{Partially funded by the French ANR projets blancs `Formal 
  Verification of Distributed Components' PiCoq ANR 2010 BLAN 0305 01 
  and `Realizability for classical logic, concurrency, references and 
  rewriting' Récré ANR-11-BS02-0010.}%
\and 
Thomas Seiller\inst{3}${}^\star$
} %
\institute{${}^{\mathrm{1}}$ ENS Cachan \ \  ${}^{\mathrm{2}}$ CNRS and Université de Savoie \ \ 
${}^{\mathrm{3}}$ INRIA}

\maketitle

\begin{abstract}
  We define a semantics for Milner's pi-calculus, with three main
  novelties. First, it provides a fully-abstract model for fair
  testing equivalence, whereas previous semantics covered variants of
  bisimilarity and the may and must testing equivalences. Second, it
  is based on reduction semantics, whereas previous semantics were
  based on labelled transition systems.  Finally, it has a strong game
  semantical flavor in the sense of Hyland-Ong and Nickau.  Indeed,
  our model may both be viewed as an innocent presheaf semantics and
  as a concurrent game semantics.
\end{abstract}

\section{Introduction}
The $\pi$-calculus~\cite{Engberg86acalculus,Milner:pi} was designed as
a basic model to reason about concurrent programs, as the
$\lambda$-calculus for functional programs. Its behavioural theory
features several notions of equivalence, including variants of
bisimilarity, contextually-defined congruences, and testing
equivalences~\cite{DBLP:journals/tcs/NicolaH84}.
Its denotational semantics has been thoroughly
investigated~\cite{DBLP:conf/lics/FioreMS96,DBLP:conf/lics/Stark96,DBLP:conf/lics/FioreT01,DBLP:conf/lics/FioreS06,Popescu09,DBLP:conf/ctcs/CattaniSW97,DBLP:journals/tcs/Hennessy02,DBLP:conf/fossacs/CrafaVY12,DBLP:journals/entcs/MontanariP95}.
This paper introduces a new denotational semantics for $\pi$ with
three main novelties.

\paragraph{{\bf Fair testing equivalence}}
First, our semantics provides a fully-abstract model for \emph{fair}
testing
equivalence~\cite{DBLP:conf/icalp/NatarajanC95,DBLP:conf/concur/BrinksmaRV95}, 
whereas previous work covers variants of bisimilarity, and the \emph{may}
and \emph{must} testing equivalences.

Originally introduced for CCS-like calculi, fair testing equivalence
reconciles the good properties of observation
congruence~\cite{Milner89} w.r.t.\ divergence, and the good properties
of previous testing equivalences~\cite{DBLP:journals/tcs/NicolaH84}
w.r.t.\ choice.  The idea is, as in most testing semantics, that two
processes are equivalent when they pass the same \emph{tests}. A
process $P$ passes the test $T$ iff their parallel composition
$P \para T$ never loses the ability of playing some special `tick'
action, even after any reduction sequence.  Fair testing is, to our
knowledge, one of the finest testing equivalences.

Cacciagrano et al.~\cite{DBLP:journals/corr/abs-0904-2340}, beyond
providing an excellent survey on fairness, adapt the definition to
$\pi$ and study approximations of it.  Their definition is not a
congruence, for essentially the same reason as for standard
bisimilarity~\cite{DBLP:journals/acta/Sangiorgi96}. We thus refine it
by allowing tests to rename channels, which yields a congruence.

\paragraph{{\bf Reductions vs.\ labelled transitions}}
A second novelty is that our model is based on the \emph{reduction}
semantics of $\pi$, while all others, to our knowledge, are based on
its standard \emph{labelled transition system} (\lts{}). The tension
between the two has been the subject of substantial
research~\cite{modularLTS}. Briefly, reduction semantics is simple and
intuitive, but it operates on equivalence classes of terms (under so-called
\emph{structural} congruence). On the other hand, designing \ltss{} is
a subtle task, rewarded by easier, more structural reasoning over
reductions.  
 \Ltss{} are generally perceived as less primitive than
reduction semantics, although they are often preferred for practical
reasons.

Most \ltss{} for $\pi$ distinguish two kinds of transitions for
output, respectively called \emph{free} and \emph{bound}.  This
distinction is crucial in all models mentioned above, which rely on
the standard \lts{} semantics, but not in our model. Actually, as
explained below in the proof sketch for Theorem~\ref{thm:1}, our model
suggests a new \lts{} for $\pi$ which does not distinguish between
free and bound output, in a way reminiscent of the carefully-crafted
\lts{} of Rathke and Soboci\'{n}ski~\cite{modularLTS}.

\paragraph{{\bf Innocent presheaves $=$ concurrent strategies}}
The third novelty of our semantics is its strong \emph{game
  semantical} flavor, in the sense of
Hyland-Ong~\cite{DBLP:journals/iandc/HylandO00} and
Nickau~\cite{DBLP:conf/lfcs/Nickau94}, whereas most previous work was
based on coalgebras, bialgebras, presheaves, event structures, or
graph rewriting.  Game semantics was designed to provide denotational
semantics for functional programming languages, but also led to
fully-abstract models for impure features like references or control
operators.  A few authors have defined game semantics for concurrent
languages, as discussed below.

A particular feature of our game is its truly `multi-player' aspect.
An immediate benefit of this is that parallel composition, usually
interpreted as a complex operation, is here just a move in the game,
allowing a player to \emph{fork} into two. On the other hand,
considering a multi-player game opens the door to undesirable
strategies, which are then ruled out by imposing an \emph{innocence}
condition, very close in spirit to game semantical
innocence~\cite{DBLP:journals/iandc/HylandO00,DBLP:conf/lfcs/Nickau94}. Indeed,
it amounts to requiring that players interact according only to their
local \emph{view} of the play.

A particular feature of our \emph{strategies} helps dealing with
$\pi$'s \emph{external choice} operator, which allows processes to
accept the same action on some channel $a$ in several ways.  E.g., any
process of the shape $a(x).P + a(x).Q$ may input on $a$ in two ways,
resp.\ leading to $P$ and $Q$.  In order to model this, we use the
fact\footnote{The second author learnt this from a talk by Sam
  Staton.}  that presheaves on the poset $P$ of plays ordered by
prefix are actually a form of concurrent strategies.  Indeed, a
strategy is traditionally a prefix-closed set of plays.  This is
equivalent to a functor $F \colon \op{P} \to \two$, where $\two$ is
the poset $0 \to 1$, viewed as a category: the accepted plays $p \in
P$ are those such that $F(p) = 1$. The action of $F$ on morphisms
ensures prefix closedness, because for any plays $p \leq q$, we must
have $F(q) \leq F(p)$; hence, if $q$ is accepted, then so is $p$. Now,
$\two$ embeds fully and faithfully into sets by $0 \mapsto \emptyset$
and $1 \mapsto 1$ (the latter denoting any singleton set
$\ens{\star}$): this is as if strategies could accept plays in one way
only, $\star$.  Presheaves generalise this by mapping to arbitrary
sets (or even just finite ones, as we do), and we think of $F(p)$ as
the set of possible ways for $F$ to accept $p$. Sticking to strategies
as sets of plays would lead to models of coarser equivalences, as
obtained by Ghica-Murawski~\cite{DBLP:conf/fossacs/GhicaM04} and
Laird~\cite{
DBLP:conf/fsttcs/Laird06}.  (Harmer and
McCusker~\cite{DBLP:conf/lics/HarmerM99}, on the other hand, consider
a finer equivalence for a non-deterministic language rather
than a concurrent one.)

So, our strategies may both be viewed as a concurrent variant of game
semantical innocent strategies, and as an innocent variant of
presheaves.

\paragraph{{\bf Playgrounds}}
Finally, our category of strategies is not constructed by hand. It is
derived using the previously defined theory of
\emph{playgrounds}~\cite{DBLP:conf/calco/Hirschowitz13}. This theory
draws direct inspiration from \emph{Kleene
  coalgebra}~\cite{DBLP:conf/fossacs/BonsangueRS09}.

In Kleene coalgebra, the main idea is that both the syntax and the
semantics of various kinds of automata should be \emph{derived} from
more basic data describing, roughly, the `rule of the game'. Formally,
starting from a well-behaved (\emph{polynomial}) endofunctor on sets,
one constructs both (1) an equational theory and (2) a sound and
complete coalgebraic semantics. This framework has been applied in
traditional automata theory, as well as in quantitative settings.
Nevertheless, its applicability to programming language theory is yet
to be established. E.g., the derived languages do not feature parallel
composition.

Playgrounds may be seen as a first attempt to convey such ideas to the
area of programming language theory.

In~\cite{DBLP:conf/calco/Hirschowitz13}, it was shown how to construct,
from any playground $\D$, (1) a syntax and a transition system
$\SSS_\D$, together with (2) a denotational model in terms of
innocent, concurrent strategies as described above.  This mimicks the
syntactic and semantic models derived in Kleene coalgebra, except
that we replace the equational theory by a transition system, and the
coalgebraic semantics by a game semantics.  Then, a playground $\Dccs$
was constructed and shown, by embedding CCS into $\SSS_{\Dccs}$, to
give rise to a fully-abstract model for fair testing equivalence.

In this paper, we construct a playground $\D$ for $\pi$ and show,
using similar techniques, that it yields a fully-abstract model for
our variant of fair testing equivalence (Theorems~\ref{thm:1}
and~\ref{thm:2}).

\paragraph{{\bf Related work}}
Building upon previous
work~\cite{DBLP:conf/lics/AbramskyM99,DBLP:conf/concur/MelliesM07} on
\emph{asynchronous} games, recent work by Winskel and
collaborators~\cite{RideauW,DBLP:conf/fossacs/Winskel13} attempts to
define a notion of concurrent game encompassing both innocent game
semantics and presheaf models. Ongoing work shows that the model does
contain innocent game semantics, but presheaf models are yet to be
investigated.  

Furthermore, our model is inspired by Girard's
\emph{ludics}~\cite{DBLP:journals/mscs/Girard01}, Melliès's game
semantics in string diagrams~\cite{DBLP:conf/lics/Mellies12}, Harmer
et al.'s categorical combinatorics of
innocence~\cite{DBLP:conf/lics/HarmerHM07}, and combinatorial
structures from algebraic topology~\cite{LeinsterHC}.

Finally, Hildebrandt's approach to fair testing
equivalence~\cite{DBLP:journals/tcs/Hildebrandt03} uses related
techniques, namely sheaves. We also use sheaves: our innocence
condition may be viewed as a sheaf condition, as briefly reviewed in
Sect.~\ref{sec:strats}.  In Hildebrandt's work, sheaves are used to
correctly handle infinite behaviour, whereas here they are used to
force reactions of players to depend only on their view.

\paragraph{{\bf Plan}} In Sect.~\ref{sec:playground}, we sketch the
construction of our playground for $\pi$, recalling the notion along
the way.  We emphasise one particular axiom for playgrounds, which was
most challenging when passing from CCS to $\pi$. We then recall the
notion of strategy.  In Sect.~\ref{sec:pi}, we recall and instantiate
the transition system for strategies constructed
in~\cite{DBLP:conf/calco/Hirschowitz13} and define our variant of fair
testing equivalence. Finally, we state our main results.

\paragraph{{\bf Perspectives}}
We plan to adapt our semantics to other calculi like the Join and
Ambients calculi, and ultimately get back to functional calculi.  We
hope to eventually generalise it, e.g., to some SOS format. More
speculative directions include (1) defining a notion of morphism for
playgrounds which would induce translations between strategies, and
find sufficient conditions for such morphisms to preserve, resp.\
reflect behavioural equivalences; (2) applying playgrounds beyond
programming language semantics; in particular, preliminary work shows
that playgrounds easily account for cellular automata, which provides
a testbed for morphisms of
playgrounds~\cite{DBLP:journals/tcs/DelormeMOT11}.

\paragraph{{\bf Notation}}
Throughout the paper, any finite ordinal $n$ is seen as $\ens{1, 
  \ldots, n}$ (rather than $\ens{0, \ldots, n-1}$). 
$\Set$ is the category of sets; $\set$ is the category of finite
ordinals and arbitrary maps; $\ford$ is the category of finite
ordinals and monotone maps. For any category $\C$, put $\Psh{\C} =
[\op\C,\Set]$, $\FPsh{\C} = [\op\C,\set]$, $\OPsh{\C} =
[\op\C,\ford]$, and let $\Chatf$ denote the category of \emph{finite}
presheaves, i.e., those presheaves $F$ such that $\sum_{c \in \ob
  (\C)} F (c)$ is finite. For all presheaves $F$ of any such kind, $x
\in F(d)$, and $f \colon c \to d$, let $x \cdot f$ denote $F(f)(x)$.

Our $\pi$-calculus processes will be infinite terms generated by
the typing rules:
\begin{mathpar}
  \inferrule{ \ldots \\ \Gam \cdot \alpha_i \vdash P_i \\ \ldots (\forall i \in n) %
    }{ %
      \Gam \vdash \sum_{i \in n} \alpha_i.P_i
    }
    \and
    \inferrule{\Gam \vdash P \\ \Gam \vdash Q}{\Gam \vdash P \para Q}
    \and
    \inferrule{\Gam + 1 \vdash P}{\Gam \vdash \nu.P}~\cdot
\end{mathpar}
$\Gam$ ranges over finite ordinals viewed as sets of variables
$\ens{1, \ldots, \Gam}$, and $\Gam \cdot \alpha$ assumes that $\alpha
::= \send{a}{b} \aalt a \aalt \tick$, for $a,b \in \Gam$. In that
case, $\Gam \cdot \sendab = \Gam \cdot \tick = \Gam$ and $\Gam \cdot a
= \Gam+1$. The $\tick$ form is a `tick' action used to define fair
testing equivalence.  This is a de Bruijn-like presentation, which we
equip with any standard reduction semantics~\cite{Milner89}, viewed
as a reflexive graph $\picalc$.  For any $\Gam \vdash P$ and map $h
\colon \Gam \to \Del$ (of finite sets), we denote by $\Del \vdash
P[h]$ the result of renaming channels according to $h$ in $P$.

Finally, and this will only be used in sketching our proof of
Theorem~\ref{thm:1}, we consider a slightly more general notion of
\lts{} than usual.  We work in the category $\Gph$ of reflexive
(directed, multi) graphs, and our category of \ltss{} over $A$ is the
slice category $\Gph / A$.  The usual notion of \anlts{} over an
alphabet $\Sigma$ is recovered by taking for $A$ the free one-vertex
reflexive graph with edges in $\Sigma$ (and by restricting to faithful
morphisms over $A$).

\section{Diagrams and plays}\label{sec:playground}
In this section, we sketch the construction of our playground for
$\pi$, and recall the notion of strategies. As sketched in the
introduction, our playground will model a multi-player game,
consisting of positions and plays between them. Positions are certain
graph-like objects, where vertices represent players and channels.
But what might be surprising is that moves are not just a binary
relation between positions, because we not only want to say
\emph{when} there is a move from one position to another, but also
\emph{how} one moves from one to the other. This will be implemented
by viewing moves from $X$ to $Y$ as \emph{cospans} $X \to M \ot Y$ in
a certain category $\Chatf$ of higher-dimensional graph-like objects,
where $X$ and $Y$ respectively are the initial and final positions,
and $M$ describes how one goes from $X$ to $Y$.  By composing such
moves (by pushout), we get a bicategory $\Dv$ of positions and
plays. We then go on and equip this bicategory with more structure,
namely that of a pseudo double category, where one direction models
dynamics, and the other models space, e.g., the inclusion of a
position into another. We then explain why the so-called `fibration'
axiom is non-obvious and describe our solution.

\subsection{Diagrams}
In preparation for the definition of our base category $\C$, recall
that (directed, multi) graphs may be seen as presheaves over the
category freely generated by the graph with two objects $\star$ and
$[1]$, and two edges $s,t \colon \star \to [1]$. Any presheaf $G$
represents the graph with vertices in $G(\star)$ and edges in $G[1]$,
the source and target of any $e \in G[1]$ being respectively $e \cdot
s$ and $e \cdot t$. A way to visualise how such presheaves represent
graphs is to compute their \emph{categories of
  elements}~\cite{MM}. Recall that the category of elements $\int G$
for a presheaf $G$ over $\C$ has as objects pairs $(c,x)$ with $c \in
\C$ and $x \in G(c)$, and as morphisms $(c,x) \to (d,y)$ all morphisms
$f \colon c \to d$ in $\C$ such that $y \cdot f = x$. This category
admits a canonical functor $\pi_G$ to $\C$, and $G$ is the colimit of
the composite $\int G \xto{\pi_G} \C \xto{\yoneda} \Chat$ with the
Yoneda embedding. E.g., the category of elements for $\yoneda[1]$ is
the poset $(\star, s) \xto{s} ([1],\id_{[1]}) \xot{t} (\star, t)$,
which could be pictured as
\diagramme[stringdiag={0.1}{0.6}]{baseline=(A.south)}{%
  \path[-,draw] %
  (A) edge (E) %
  (B) edge (E) %
  ; %
  \node at ($(B.south east) + (.1,0)$) {,} ;%
}{%
  \joueur{A} \& \node[regular polygon,anchor=center,regular polygon
  sides=3,fill,minimum size=3pt,draw,rotate=-90] (E) {}; \&
  \joueur{B} %
}{%
} \hspace*{-.7em} where dots represent vertices, the triangle
represents the edge, and links materialise the graph of $G(s)$ and
$G(t)$, the convention being that $t$ goes from the apex of the
triangle.  We thus recover some graphical intuition.

Our string diagrams will also be defined as (finite) presheaves over
some base category $\C$. Let us give the formal definition of $\C$ for
reference.  We advise to skip it on a first reading: we then attempt
to provide some graphical intuition.
\begin{figure}[t]
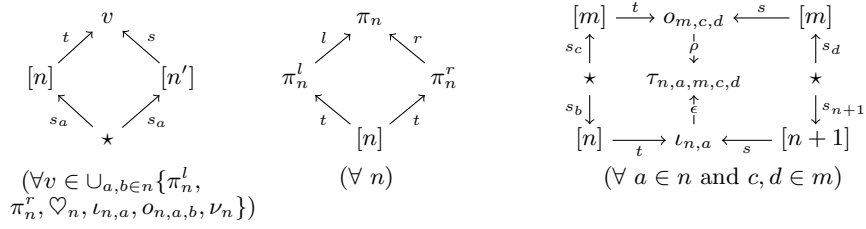

    \begin{mathpar}
      \begin{minipage}[t]{0.22\textwidth}
        \centering
        {\diag(.4,.4){%
            \&|(v)| v \\
            |(n)| [n] \& \& |(n')| [n'] \\
            \& |(star)| \star %
          }{%
            (star) edge[labelbl={s_a}] (n) %
            edge[labelbr={s_a}] (n') %
            (n) edge[labelal={t}] (v) %
            (n') edge[labelar={s}] (v) %
          }} \\
        ($\forall v \in {\cup_{a,b \in n}} \{\forkln,\linebreak \forkrn,
        \tickn, \inna, \outnab, \nun\}$)
      \end{minipage}
  \and
      \begin{minipage}[t]{0.22\textwidth}
        \centering
        {\diag(.4,.4){%
      \&|(v)| \forkn \\
     |(n)| \forkln \& \& |(n')| \forkrn \\
      \& |(star)| [n] %
    }{%
      (star) edge[labelbl={t}] (n) %
       edge[labelbr={t}] (n') %
       (n) edge[labelal={l}] (v) %
       (n') edge[labelar={r}] (v) %
    }} \\
  ($\forall$ $n$)
  \end{minipage}
  \and 
  \begin{minipage}[t]{0.44\textwidth}
    \centering
    {\diag(.4,.3){%
        |(n)| [m]   \&  |(sender)| \outmcd \& |(ni)| [m] \\
        |(star)| \star \& |(v)| \taunamcd \& |(stari)| \star \\
        |(n')| [n] \& |(receiver)| \inna \& |(n'i)| [n+1] %
        \& %
      }{%
        (star) edge[labell={s_c}] (n) %
        edge[labell={s_b}] (n') %
        (n) edge[labela={t}] (sender) %
        (n') edge[labelb={t}] (receiver) %
        (sender) edge[labelo={\rho}] (v) %
        (receiver) edge[labelo={\epsilon}] (v) %
        (stari) edge[labelr={s_d}] (ni) %
        edge[labelr={s_{n+1}}] (n'i) %
        (ni) edge[labela={s}] (sender) %
        (n'i) edge[labelb={s}] (receiver) %
      }} \\
    ($\forall$ $a \in n$ and $c,d \in m$)
    \end{minipage}%
  \end{mathpar}%
  \caption{Equations for $\C$}
  \label{fig:equationsC}
\end{figure}
\begin{definition}
  Let $G_{\C}$ be the graph with, for all $n$, $m$, with $a,b \in n$ and $c,d \in m$:
  \begin{itemize}
  \item vertices $\star$, $[n]$, $\forkln$, $\forkrn$, $\forkn$,
    $\nun$, $\tickn$, $\inna$, $\outnab$, and $\taunamcd$;
  \item edges $s_1,...,s_n : \star \to [n]$;
  \item for all $v \in \ens{\forkln,\forkrn,\tickn,\outnab}$, edges
    $s,t : [n] \to v$;
  \item edges $[n] \xto{t} \nun \xot{s} [n+1]$ and $[n] \xto{t} \inna
    \xot{s} [n+1]$;
  \item edges $\forkln \xto{l} \forkn \xot{r} \forkrn$;
  \item edges $\inna \xto{\rho} \taunamcd \xot{\epsilon} \outmcd$.
  \end{itemize}

  Let $\C$ be the free category on $G_{\C}$, modulo the equations in
  Fig.~\ref{fig:equationsC}, where, in the left-hand one, $n'$ is $n+1$
  when $v = \nun$ or $\inna$, and $n$ otherwise.
\end{definition}
Our category of string diagrams will be the category of finite
presheaves $\Chatf$.

\begin{wrapfigure}{r}{0pt}
  \begin{minipage}[c][3em]{0.33\linewidth}
    \vspace*{-.2em}
      \diagramme[stringdiag={.8}{1.3}]{}{%
}{%
  \node (s_1) {$(\star, s_1)$}; \& \node (s_2) {$(\star, s_2)$}; \& \node (s_3) {$(\star, s_3)$}; \\
    \& \node (id) {$([3], \id_{[3]})$}; 
    }{%
      (s_1) edge (id) %
      (s_2) edge (id) %
      (s_3) edge (id) %
    }
\end{minipage}
\end{wrapfigure}
To explain this seemingly arbitrary definition, let us compute a few
categories of elements. Let us start with an easy one, that of $[3]
\in \C$ (we implicitly identify any $c \in \C$ with $\yoneda c$). An
easy computation shows that it is the poset pictured in the top part
on the right. We will think of it as a position with one

\begin{wrapfigure}{r}{0pt}
      \diagramme[stringdiag={.8}{1.3}]{}{%
    \path[-,draw] %
    (a) edge (j1) %
    (c) edge (j1) %
    (b) edge (j1) %
    ; %
}{%
    \canal{a}     \& \canal{b} \&  \canal{c} \\
    \& \joueur{j1}
    }{%
    }
\end{wrapfigure}
\noindent 
player $([3],\id_{[3]})$ connected to three channels, and draw it as
in the bottom part on the right, where the bullet represents the
player, and circles represent channels. 
In particular, elements over $[3]$ represent ternary
players, while elements over $\star$ represent channels.  The
\emph{positions} of our game are finite presheaves empty except
perhaps on $\star$ and $[n]$'s. Other objects will represent moves.
The graphical representation
is slightly ambiguous, because the ordering of channels known to
players is implicit.  We will disambiguate in the text when
necessary. 

A more difficult category of elements is that of $\paraof{2}$. It is
the poset generated by the graph on the left (omitting the base
objects for conciseness):
  \begin{center}
    \diag (.4,.3) {%
      \& \& |(lt)| l s \& \& |(rt)| r s \& \& \\ 
      |(lt1)| l s s_1 \& \& |(l)| l \& |(para)| \id_{\paraof{2}} \& |(r)| r \& \& |(lt2)| l s s_2 \\ 
      \& \& \& |(ls)| l t = r t \&  \& \& 
    }{%
      (lt1) 
      edge (ls) %
      (lt2) 
      edge (ls) %
      (ls) edge (l) edge (r) %
      (lt) edge (l) %
      (rt) edge (r) %
      (l) edge (para) %
      (r) edge (para) %
      (lt1) edge[identity] (lt1) 
      edge (lt) %
      edge[fore,bend left=10] (rt) %
      (lt2) edge[identity] (lt2) 
      edge (rt) %
      edge[bend right=10,fore] (lt) %
    }
    \hfil
          \diagramme[stringdiag={.3}{.6}]{}{
    \node[diagnode,at= (t1.south east)] {\ \ \ .} ; %
  }{%
     \& \& \joueur{t_1} \&  \& \joueur{t_2} \\
    \& \&   \&  \\
    \& \ \& \\
    \canal{t0} \& \& \& \couppara{para} \& \& \& \canal{t1} \\ 
    \& \ \& \\
    \& \&  \\
    \& \& \& \joueur{s} \& \& \& \& 
  }{%
    (para) edge[-] (t_2) %
    (t1) edge[-,bend right=10] (t_2) %
    (t0) 
    edge[-] (s) %
    (t1) 
    edge[-] (s) %
    (s) edge[-] (para) %
    (para) edge[fore={.3}{.3},-] (t_1)
    (t0) edge[fore={.5}{.5},-,bend left=10] (t_2) %
    (t0) edge[-,bend left=15] (t_1) %
    (t1) edge[-,fore={1}{.5},bend right=10] (t_1) %
  }  
  \end{center}
  We think of it as a binary player ($l t$) forking into two players
  ($l s$ and $r s$), and draw it as on the right. 
  \newcommand{\longueurfigun}{.6} \newcommand{\separation}{} The
  graphical convention is that a black triangle stands for the
  presence of $\id_{\forkof{2}}$, $l$, and $r$. Below, we
  represent just $l$ as a white triangle with only a left-hand branch,
  and symmetrically for $r$.  Furthermore, in all our pictures, time
  flows `upwards'. 

  Another category of elements, characteristic of the $\pi$-calculus,
  is the one for synchronisation $\taunamcd$. The case $(n,a,m,c,d) =
  (1,1,3,2,3)$ is the poset generated by the graph on the left of
  Fig.~\ref{fig:tau}, which we will draw as on the right. %
  \begin{figure}[t]
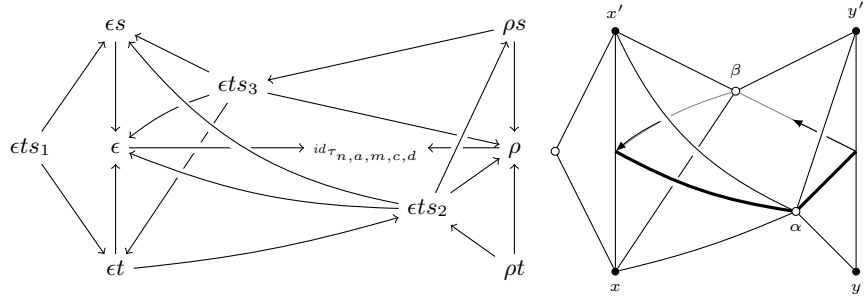

    \begin{mathpar}
      \diagramme[diag={.4}{.5}]{}{%
        \path[->] %
        (s) edge[bend right=5] (t1) %
        (t2) edge (s) %
        (t0) edge (t) %
        (t0) edge (s) %
        (t2) edge (t) %
        (s') edge (t1) %
        (t') edge (t2) %
        (t2) edge[->] (iota') %
        ; %
        \path[->,draw] %
        (t1) edge[fore={.3}{.5},bend left=10] (iota) %
        (t1) edge[fore={0.3}{.5}] (iota') %
        ; %
        \path[draw,->] (t2) edge[->,bend right=10] (iota) %
        (iota) edge[fore={.4}{0}] (tau) %
        (iota') edge[fore={.4}{0}] (tau) %
        ; %
        \path[->] (t1) edge[fore={.3}{.3}] (t') %
        ; %
        \path[->] (t1) edge[fore={.1}{.5},bend left=20] (t) ; %
        \foreach \x/\y in {s/iota,t/iota,s'/iota',t'/iota'} \path[->]
        (\x) edge (\y) ; %
      }{%
        \&  |(t)| \epsilon s \& \&   \& \&  |(t')| \rho s   \\
        \&   \& \& |(t2)| \epsilon t s_3 \\
        |(t0)[anchor=base west]| \epsilon t s_1 \&  |(iota)| \epsilon \& \& \& |(tau)[anchor=base east]| {\scriptscriptstyle \id_{\taunamcd}} \&  |(iota')| \rho \\
        \&  \& \& \& |(t1)| \epsilon t s_2 \\
        \& |(s)| \epsilon t \& \& \& \& |(s')| \rho t 
      }{%
      }%
      \diagramme[stringdiag={.8}{.8}]{}{%
        \path[-] %
        (s) edge[bend right=5] (t1) %
        (t2) edge (s) %
        (t0) edge (t) %
        (t0) edge (s) %
        (t2) edge (t) %
        (s') edge (t1) %
        (t') edge (t2) %
        (iota') edge[gray,very thin] (t2) %
        ; %
        \path[-] %
        (iota) edge[fore={.3}{.5},very thick,bend right=10] (t1) %
        (iota') edge[fore={0.3}{.5},very thick] (t1) %
        ; %
        \path[] (t2) edge[color=gray,very thin,bend right=10]
        node[coordinate,pos=.6] (iotatip) {} (iota) ; %
        \path (iota') -- (t2) node[coordinate,pos=0.55] (iotatip') {}
        ; %
        \path[-] (iotatip) edge[bend right=7,-latex] (iota) ; %
        \path[-] (iota') edge[-latex] (iotatip') ; %
        \path[-] (t1) edge[fore={.3}{.3}] (t') %
        ; %
        \path[-] (t) edge[fore={.1}{.5},bend right=20] (t1) ; %
        \foreach \x/\y in {s/t,s'/t'} \path[-] (\x) edge (\y) ; %
        \node[anchor=south] at (t2.north) {$\scriptstyle \beta$} ; %
        \node[anchor=north] at (t1.south) {$\scriptstyle \alpha$} ; %
        \node[anchor=south] at (t.north) {$\scriptstyle x'$} ; %
        \node[anchor=north] at (s.south) {$\scriptstyle x$} ; %
        \node[anchor=south] at (t'.north) {$\scriptstyle y'$} ; %
        \node[anchor=north] at (s'.south) {$\scriptstyle y$} ; %
        \node[diagnode,at= (s'.south east)] {\ \ \ .} ; %
      }{%
        \& \joueur{t} \& \& \& \& %
        \joueur{t'}   \\
        \&   \& \& \canal{t2} \\
        \canal{t0} \&  \coupout{iota}{0} \& \&   \& \&  \coupin{iota'}{0} \\
        \&  \& \& \& \canal{t1} \\
        \& \joueur{s} \& \& \& \& \joueur{s'} 
      }{%
      }%
    \end{mathpar}
    \caption{Category of elements for $\tau_{1,1,3,2,3}$ and graphical representation}
\label{fig:tau}
\end{figure} %
The left-hand ternary player $x$ outputs its $3$rd channel, here $\beta$,
on its $2$nd channel, here $\alpha$. The right-hand unary player $y$
receives the sent channel on its $1$st channel, here $\alpha$. The carrier
channel is marked with thick lines, while the transmitted channel is
indicated with arrows.  Both players have two occurrences, one before
and one after the move, respectively marked as $x / x'$ and $y / y'$.
Both $x$ and $x'$ have arity $3$ here, while $y$ has arity $1$ and
$y'$, having gained knowledge of channel $\beta$, has arity
$2$. 

We leave the computation of other categories of elements as an
exercise to the reader. 
The remaining diagrams for $\paralp$, $\pararp$, $\outmcd$, $\inna$,
$\tickp$, and $\nup$ are depicted below, for $p = 2$ and
$(m,c,d,n,a) = (3,2,3,1,1)$: \\ 
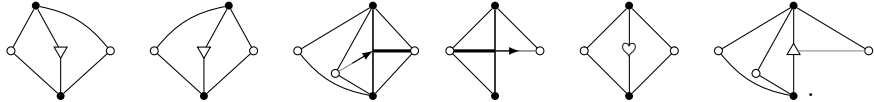
\noindent
\begin{minipage}[t][5em]{\textwidth}
    \diagramme[stringdiag={.2}{.33}]{}{ }{%
       \& \joueur{t_1} \& \& \& 
      \\ 
      \& \&   \&  \\
      \& \ \& \\
      \canal{t0} \& \& \coupparacreux{para} \& \& \canal{t1}
      \\ 
      \& \ \& \\
      \& \&  \\
       \& \& \joueur{s} \& \&
    }{%
      (t0) edge[-] (t_1) %
      (t1) edge[-,bend right=20] (t_1) %
      (t0) edge[-] (s) %
      (t1) edge[-] (s) %
      (s) edge[-] (para) %
      (para) edge[-] (t_1) %
    }
    \separation
%
    \diagramme[stringdiag={.2}{.33}]{}{ }{%
       \& \& \& \joueur{t_2} \& 
      \\ 
      \& \&   \\
      \& \ \& \\
      \canal{t0} \& \& \coupparacreux{para}
      \& \& \canal{t1}
      \\ 
      \& \ \& \\
      \& \&  \\
       \& \& \joueur{s} \& \&
    }{%
      (t0) edge[-,bend left=20] (t_2) %
      (t1) edge[-] (t_2) %
      (t0) edge[-] (s) %
      (t1) edge[-] (s) %
      (s) edge[-] (para) %
      (para) edge[-] (t_2) %
    }
    \separation
%
    \diagramme[stringdiag={.3}{.5}]{baseline=($(iota.south)$)}{%
      \path[-] %
      (t2) edge (s) %
      (t1) edge (t) %
      (t0) edge (t) %
      (t2) edge (t) %
      (t) edge (iota.west) %
      (s) edge (iota.west) %
      (t0) edge[bend right=20] (s) %
      (t1) edge (s) %
      ; %
      \movepiout[1]{t1}{iota}{t2}{.4} %
      \foreach \x/\y in {s/t} \path[-] (\x) edge (\y) ; %
    }{%
      \& \& \joueur{t}  \\
      \&  \\
      \canal{t0} \& \& \coupout{iota}{0}  \& \canal{t2} \\ 
      \& \canal{t1} \\ 
      \& \& \joueur{s} 
    }{%
    }%
    \separation
    \diagramme[stringdiag={.6}{\longueurfigun}]{baseline=($(in.south)$)}{
      \path[-] (a) 
      edge (p) %
      (in) edge (p) edge (p') %
      edge[gray,very thin] (b) %
      (p') edge (a) edge (b) %
      ; %
      \movepiin[.5]{a}{in}{b}{.6} %
      \foreach \x/\y in {p/p',a/a} \path[-] (\x) edge (\y) ; %
    }{ %
       \& \joueur{p'} \&  \\ %
      \canal{a} \& \coupout{in}{0} \& \canal{b} \\ %
       \& \joueur{p} 
    }{%
    } %
    \separation
%
    \diagramme[stringdiag={.6}{\longueurfigun}]{}{ \path[-] (a) edge
      (a) %
      edge (p) %
      (tick) edge[shorten <=-1pt] (p) edge[shorten <=-1pt] (p') %
      (p') edge (a) edge (b) %
      (b) edge (p) edge (b) %
      ; %
    }{ %
       \& \joueur{p'} \& \\ %
      \canal{a} \& \couptick{tick} \& \canal{b} \\ %
       \& \joueur{p} \&  %
    }{%
    } \separation
    %
    \diagramme[stringdiag={.3}{.5}]{baseline=($(nu.center)$)}{%
      \path[-,draw] %
      (t1) edge (s) %
      (t1) edge (t) %
      (t0) edge (t) %
      (t2) edge (t) %
      (t) edge (nu) %
      (s) edge (nu) %
      (nu) edge[gray,very thin] (t2) %
      (t0) edge[bend right=20] (s) %
      ; %
      \node[diagnode,at= (s.base east)] {\ \ \ .} ; %
    }{%
      
      \& \& \joueur{t} \& \&  \& \\
      \&  \&    \\
      \canal{t0} \& \& \coupnu{nu} \& \& \canal{t2}  \\
      \& \canal{t1} \\
      \& \& \joueur{s} \& 
    }{%
    }%
  \end{minipage}
  The first two are \emph{views}, in the game semantical sense, of the
  fork move $\forkof{2}$ explained above. The next two, $\outmcd$ (for
  `output') and $\inna$ (for `input'), respectively represent what the
  sender and receiver can see of the above synchronisation move. The
  next diagram is a `tick' move, used for defining fair testing
  equivalence. The last one is a channel creation move.

\subsection{From diagrams to moves}
In the previous section, we have defined our category of diagrams as
$\Chatf$, and provided some graphical intuition on its objects.  The
next goal is to construct a bicategory whose objects are positions (recall:
presheaves empty except perhaps on $\star$ and $[n]$'s), and whose
morphisms represent plays in our game. We start in this section by
defining moves, and continue in the next one by explaining how to
compose moves to form plays. Moves are defined in two stages:
\emph{seeds}, first, give the local form for moves; moves are then
defined by embedding seeds into bigger positions.

To start with, until now, our diagrams contain no information about
the `flow of time'. To add this information, for each diagram $M$
representing a move, we define its initial and final positions, say
$X$ and $Y$, and view the whole move as a cospan $Y \xto{s} M
\xot{t} X$. We have taken care, in drawing our diagrams before, of
placing initial positions at the bottom, and final positions at the
top.  So, e.g., the initial position $X$ and final position $Y$ for
the synchronisation move are pictured on the right and they map into
(the representable
\begin{wrapfigure}[2]{r}{0pt}
    \diagramme[stringdiag={.3}{.5}]{baseline=($(nu.center)$)}{%
      \path[-,draw] %
      (t) edge (t0) %
      edge (a) %
      edge (b) %
      (t') edge (a) %
      ; %
    }{%
      \& \& \canal{b} \\
      \canal{t0} \& \joueur{t} \& \& \joueur{t'} \\
      \& \& \canal{a} \\
    }{%
    }%
     \ {$\leadsto$} \ 
    \diagramme[stringdiag={.3}{.5}]{baseline=($(nu.center)$)}{%
      \path[-,draw] %
      (t) edge (t0) %
      edge (a) %
      edge (b) %
      (t') edge (a) %
      edge (b) %
      ; %
    }{%
      \& \& \canal{b} \\
      \canal{t0} \& \joueur{t} \& \& \joueur{t'} \\
      \& \& \canal{a} \\
    }{%
    }%
\end{wrapfigure}
presheaf over) $\tau_{1,1,3,2,3}$ in the
obvious ways, yielding the cospan $Y \xto{s} M \xot{t} X$.
We leave it to the reader to define, based on the above pictures, the expected cospans
\begin{center}
  \Diag(.6,.2){%
    \foreach \i in {1,...,8} %
    \path[->,draw] %
    (m-1-\i) edge (m-2-\i) %
    (m-3-\i) edge (m-2-\i) %
    ; %
    }{%
    {[n] \para [n]} \& {[m] \paraofij{c,d}{a,n+1} [n+1] }  \&  { [n] }  \&  { [n] }  \&  { [m] }  \&  
    { [n+1] }  \&  { [n] }  \&  { [n+1] } \\
    { \forkn }  \&  { \taunamcd }  \&  { \forkln }  \&  { \forkrn }  \&  { \outmcd }  \&  
  { \inna }  \&  { \tickn }  \&  { \nun } \\
    { [n] }  \&  { [m] \paraofij{c}{a} [n] }  \&  { [n] }  \&  { [n] }  \&  { [m] }  \&  
  { [n] }  \&  { [n] }  \&  { [n], } %
  }{%
  }
\end{center}
where initial positions are on the bottom row, and we denote by
$[m] \paraofij{a_1,\ldots,a_p}{c_1,\ldots,c_p} [n]$ the position
consisting of an $m$-ary player $x$ and an $n$-ary player $y$,
quotiented by the equations $x \cdot s_{a_k} = y \cdot s_{c_k}$ for
all $k \in p$. When both lists are empty, by convention, $m=n$
and the players share all channels in order.
\begin{definition}
  These cospans are called \emph{seeds}. Their lower legs are
  called \emph{t-legs}.
\end{definition}

As announced, the moves of our game are obtained by embedding seeds
into bigger positions. This means, e.g., allowing a fork move to occur
in a position with more than one player. We proceed as follows.
\begin{definition}\label{def:interface}
  Let the \emph{interface} of a seed $Y \xto{s} M \xot{t} X$ be
  $I_X = X(\star) \cdot \star$, i.e., the position consisting only of
  the channels of the initial position of the seed.
\end{definition}

\begin{wrapfigure}{r}{0pt}
  \begin{minipage}[c][3em]{0.33\linewidth}
        \diag(.3,.6){%
    \&|(I)| I_X \\
   |(X)| Y \&|(M)| M \&|(Y)| X %
  }{%
    (I) edge (X) edge (M) edge (Y) %
    (X) edge (M) %
    (Y) edge (M) %
  }
  \end{minipage}
\end{wrapfigure}
Since channels present in the initial position remain in the final
one, we have for each seed a commuting diagram as on the right.  By gluing any
position $Z$ to the seed along its interface, we obtain a new cospan,
say $Y' \to M' \ot X'$.  I.e., for any morphism $I_X \to Z$, we push
$I_X \to X$, $I_X \to M$, and $I_X \to Y$ along $I_X \to Z$ and use
the universal property of pushout, as in:
  \begin{center}
    \Diag(.3,.8){%
      \pbk{X}{X'}{Z} %
      \pullback{Z}{M'}{M}{draw,-} %
      \pullback{Z}{Y'}{Y}{draw,-} %
    }{%
      \& |(Y)| Y \& \& |(Y')| {Y'} \\
      \& \ \& \\
      \& |(M)| M \& \& |(M')| M'  \\
      |(I)| I_X \&\& |(Z)| Z \\
      \& |(X)| X \& \& |(X')| X'. %
    }{%
      (Z) edge[] (X') %
      edge (M') %
      edge (Y') %
      (I) edge[] (X) %
      edge (Z) %
      edge (M) %
      edge (Y) %
      (Y) edge[fore] (Y') %
      (M) edge[fore] (M') %
      (X) edge (X') %
      (X') edge[dashed] (M') %
      (Y') edge[dashed] (M') %
      (X) edge[fore] (M) %
      (Y) edge[fore] (M) %
    }
  \end{center}
  \begin{definition}
    Let \emph{(global) moves} be all cospans obtained in this way.
  \end{definition}
  Recall that colimits in presheaf categories are pointwise. So, e.g.,
  taking pushouts along injective maps graphically corresponds to
  gluing diagrams together. Let us do a few examples.
  \begin{example}\label{ex:forkmove}
    The copsan $[2]\para[2] \xto{[ls,rs]} \forkof{2} \xot{lt} [2]$ has
    as canonical interface the presheaf $I_{[2]} = 2 \cdot \star$,
    consisting of two channels, say $a$ and $b$.  Consider the
    position $[2] + \star$ consisting of a player $y$ with two
    channels $b'$ and $c$, plus an additional channel $a'$. Further
    consider the map $h \colon I_{[2]} \to [2]+\star$ defined by $a
    \mapsto a'$ and $b \mapsto b'$. The pushout
    \begin{mathpar}
      \Diag{%
        \pbk{pi}{M'}{star} %
        }{%
         |(I2)| {I_{[2]}} \& |(star)| {[2]+\star} \\
         |(pi)| {\forkof{2}} \&|(M')| {M'} %
        }{%
          (I2) edge (star) edge (pi) %
          (pi) edge (M') %
          (star) edge (M') %
        } %
        \and \mbox{is} \and
                \diagramme[stringdiag={.3}{.6}]{}{
    \node[diagnode,at= (c.south east)] {\ \ \ .} ; %
        \node[anchor=south] at (t_1.north) {$\scriptstyle x_1$} ; %
        \node[anchor=south] at (t_2.north) {$\scriptstyle x_2$} ; %
        \node[anchor=north] at (s.south) {$\scriptstyle x$} ; %
        \node[anchor=north] at (y.south) {$\scriptstyle y$} ; %
        \node[anchor=north] at (c.south) {$\scriptstyle c$} ; %
        \node[anchor=north] at (t0.south) {$\scriptstyle a=a'$} ; %
        \node[anchor=north] at (t1.south) {$\scriptstyle b=b'$} ; %
  }{%
     \& \& \joueur{t_1} \&  \& \joueur{t_2} \\
    \& \&   \&  \\
    \& \ \& \\
    \canal{t0} \& \& \& \couppara{para} \& \& \& \canal{t1} \& \& \joueur{y} \& \& \canal{c} \\ 
    \& \ \& \\
    \& \&  \\
    \& \& \& \joueur{s} \& \& \& \& 
  }{%
    (para) edge[-] (t_2) %
    (t1) edge[-,bend right=10] (t_2) %
    (t0) 
    edge[-] (s) %
    (t1) 
    edge[-] (s) %
    (s) edge[-] (para) %
    (y) edge[-] (t1) %
     edge[-] (c) %
    (para) edge[fore={.3}{.3},-] (t_1)
    (t0) edge[fore={.5}{.5},-,bend left=10] (t_2) %
    (t0) edge[-,bend left=15] (t_1) %
    (t1) edge[-,fore={1}{.5},bend right=10] (t_1) %
  }  
    \end{mathpar}
  \end{example}
  \begin{example}\label{ex:interface}
    The canonical interface, being the interface of the initial
    position, may not contain all channels of the move. In particular,
    for an input move which is not part of any synchronisation, the
    received channel cannot be part of the initial position.
  \end{example}

\subsection{From moves to plays}
\begin{wrapfigure}[5]{r}{0pt}
    \diag(.3,1){%
      \&|(U)| U \\
     |(X)| X \& \&|(Y)| Y \\
      \&|(V)| V %
    }{%
      (X) edge (U) edge (V) %
      (Y) edge (U) edge (V) %
      (U) edge (V) %
    }
\end{wrapfigure}
Having defined moves, we now define their composition to define our
bicategory $\Dv$ of positions and plays.  $\Dv$ will be a
sub-bicategory of $\Cospan{\Chatf}$, the bicategory which has as
objects all finite presheaves on $\C$, as morphisms $X \to Y$ all
cospans $X \to U \ot Y$, and as 2-cells $U \to V$ all commuting
diagrams as on the right.  Composition is given by pushout, and hence
is not strictly associative. \begin{remark} We choose to view the
  initial position as the \emph{target} of the morphism in
  $\Cospan{\Chatf}$, in order to emphasise below that the fibration
  axiom is very close to a universal property of
  pullback~\cite{Jacobs}.
\end{remark}

\begin{definition}
  \emph{Plays} are composites of moves in $\Cospan{\Chatf}$.  Let
  $\Dv$ be the sub-bicategory consisting of positions and plays.
\end{definition}
\begin{remark}
We do not yet specify what the 2-cells of $\Dv$ are. This will follow from the next section.  
\end{remark}

Intuitively, composition by pushout glues
diagrams on top of each other, which features some concurrency. 
\begin{example}
  Composing the move of Example~\ref{ex:forkmove} with a forking
  move by $y$ yields
    \begin{center}
                \diagramme[stringdiag={.3}{.6}]{}{
    \node[diagnode,at= (c.south east)] {\ \ \ .} ; %
        \node[anchor=south] at (t_1.north) {$\scriptstyle x_1$} ; %
        \node[anchor=south] at (t_2.north) {$\scriptstyle x_2$} ; %
        \node[anchor=south] at (y_1.north) {$\scriptstyle y_1$} ; %
        \node[anchor=south] at (y_2.north) {$\scriptstyle y_2$} ; %
        \node[anchor=north] at (s.south) {$\scriptstyle x$} ; %
        \node[anchor=north] at (y.south) {$\scriptstyle y$} ; %
        \node[anchor=north] at (c.south) {$\scriptstyle c$} ; %
        \node[anchor=north] at (t0.south) {$\scriptstyle a=a'$} ; %
        \node[anchor=north] at (t1.south) {$\scriptstyle b=b'$} ; %
  }{%
     \& \& \joueur{t_1} \&  \& \joueur{t_2} \& \& \& \& \joueur{y_1} \& \& \joueur{y_2} \\
    \& \&   \&  \\
    \& \ \& \\
    \canal{t0} \& \& \& \couppara{para} \& \& \& \canal{t1} \& \&  \& \couppara{para'} \&  \& \& \canal{c} \\
    \& \ \& \\
    \& \&  \\
    \& \& \& \joueur{s} \& \& \& \& \& \& \joueur{y} 
  }{%
    (t1) edge[-,bend right=10] (t_2) %
    (t0) 
    edge[-] (s) %
    (t1) 
    edge[-] (s) %
    (para) edge[-] (t_2) %
    (s) edge[-] (para) %
    (y) edge[-] (t1) %
     edge[-] (c) %
    (para) edge[fore={.3}{.3},-] (t_1)
    (t0) edge[fore={.5}{.5},-,bend left=10] (t_2) %
    (t0) edge[-,bend left=15] (t_1) %
    (t1) edge[-,fore={1}{.5},bend right=10] (t_1) %
    (c) edge[-,bend right=10] (y_2) %
    (para') edge[-] (y_2) %
    (y) edge[-] (para') %
    (y) edge[-] (t1) %
     edge[-] (c) %
    (para') edge[fore={.3}{.3},-] (y_1) %
    (t1) edge[fore={.5}{.5},-,bend left=10] (y_2) %
    (t1) edge[-,bend left=15] (y_1) %
    (c) edge[-,fore={1}{.5},bend right=10] (y_1) %
  }  
    \end{center}
\end{example}
\begin{example}
  Composition retains causal dependencies between moves. To see this,
  consider the following diagram.  In the initial position, there are
  channels $a,b$, and $c$, and three players $x(a,b), y(b)$, and
  $z(a,c)$ (we indicate the channels known to each player in
  parentheses). In a first move, $x$ sends $a$ on $b$, which is
  received by $y$. In a second move, $z$ sends $c$ on $a$, which is
  received by (the avatar $y'$ of) $y$. The second move is enabled by
  the first one, by which $y$ gains knowledge of $a$. The
  corresponding diagram looks like the following, identifying the two
  framed nodes and the two circled ones:
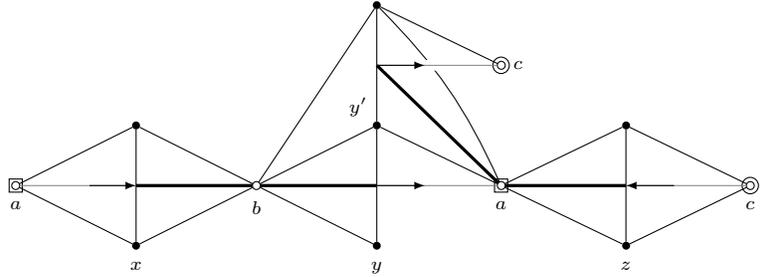
\begin{center}
      \diagramme[stringdiag={.8}{1.6}]{}{%
        \framenode{a} %
        \framenode{a'} %
        \circlenode{c} %
        \circlenode{c'} %
        \receive{b}{i}{a'}{.4}{} %
        \envoie{a}{o}{b}{.4}{} %
        \envoie{c}{o'}{a'}{.4}{} %
        \receive[fore={.1}{.1}]{a'}{i'}{c'}{.4}{} %
        \node[diagnode,at= (c.south east)] {\ \ \ \ .} ; %
        \node[below=1pt] at (a.south) {$\scriptstyle a$} ; %
        \node[below=1pt] at (b.south) {$\scriptstyle b$} ; %
        \node[below=1pt] at (c.south) {$\scriptstyle c$} ; %
        \node[below=1pt] at (a'.south) {$\scriptstyle a$} ; %
        \node[right=1pt] at (c') {$\scriptstyle c$} ; %
        \node[below=1pt] at (x.south) {$\scriptstyle x$} ; %
        \node[below=1pt] at (y.south) {$\scriptstyle y$} ; %
        \node[below=1pt] at (z.south) {$\scriptstyle z$} ; %
        \node[above left] at (y') {$\scriptstyle y'$} ; %
      }{%
        \& \& \& \joueur{y''} \\
        \& \& \& \coupin{i'}{0} \& \canal{c'} \\
        \& \joueur{x'} \& \& \joueur{y'} \& \& \joueur{z'} \\
    \canal{a} \& \coupout{o}{0} \& \canal{b} \& \coupin{i}{0} \& \canal{a'} \& \coupout{o'}{0} \& \canal{c} \\
    \& \joueur{x} \& \& \joueur{y} \& \& \joueur{z} %
  }{%
    (a) edge[-] (x) %
    edge[-] (x') %
    (b) edge[-] (x) %
    edge[-] (x') %
    edge[-] (y) %
    edge[-] (y') %
    edge[-] (y'') %
    (a') edge[-] (y') %
    edge[-,bend right=10] (y'') %
    edge[-] (z) %
    edge[-] (z') %
    (c) edge[-] (z) %
    edge[-] (z') %
    (c') edge[-] (y'') %
    (i) edge[-] (y) %
    edge[-] (y') %
    (i') edge[-] (y') %
    edge[-] (y'') %
    (o) edge[-] (x) %
    edge[-] (x') %
    (o') edge[-] (z) %
    edge[-] (z') %
  }
\end{center}  
\end{example}

\subsection{A pseudo double category}
\begin{wrapfigure}[6]{r}{0pt}
  \begin{minipage}[t][3em]{0.33\linewidth}
    \vspace*{-.6em}
  \Diag(.6,.8){%
  }{%
    X \& X' \& X'' \\
    Y \& Y' \& Y'' \\
    Z \& Z' \& Z''%
  }{%
    (m-1-1) edge[labelu={h}] (m-1-2) %
    edge[pro,labell={u},twoleft={ur}{}] (m-2-1) %
    (m-2-1) edge[labelu={h'}] (m-2-2) %
    (m-1-2) edge[pro,labell={u'},twoleft={u'r}{},tworight={u'l}{}] (m-2-2) %
    (m-1-2) edge[labelu={k}] (m-1-3) %
    (m-2-2) edge[labelu={k'}] (m-2-3) %
    (m-1-3) edge[pro,labelr={u''},tworight={u''l}{}] (m-2-3) %
    (m-2-1) 
    edge[pro,labell={v},twoleft={vr}{}] (m-3-1) %
    (m-3-1) edge[labela={h''}] (m-3-2) %
    (m-2-2) edge[pro,labell={v'},twoleft={v'r}{},tworight={v'l}{}] (m-3-2) %
    (m-3-2) edge[labela={k''}] (m-3-3) %
    (m-2-3) edge[pro,labelr={v''},tworight={v''l}{}] (m-3-3) %
    (ur) edge[cell={.2},labela={\alpha}] (u'l) %
    (u'r) edge[cell={.2},labela={\alpha'}] (u''l) %
    (vr) edge[cell={.2},labela={\beta}] (v'l) %
    (v'r) edge[cell={.2},labela={\beta'}] (v''l) %
  }
  \end{minipage}
\end{wrapfigure}
We now continue the construction of our playground for $\pi$ by adding
a new dimension. Namely, we view $\Dv$ as the vertical part of a
\emph{(pseudo) double category}~\cite{GrandisPare,GarnerPhD}.  This is
a weakening of Ehresmann's double categories~\cite{Ehresmann:double2},
where one direction has non-strictly associative composition. A pseudo
double category $\D$ consists of a set $\ob (\D)$ of \emph{objects},
shared by a `horizontal' category $\Dh$ and a `vertical' bicategory
$\Dv$. Following Paré~\cite{PareYoneda}, $\Dh$, being a mere category,
has standard notation (normal arrows and $\rond$ for composition),
while the bicategory $\Dv$ earns fancier notation ($\proto$ arrows and
$\vrond$ for composition). $\D$ is furthermore equipped with a set of
\emph{double cells} $\alpha$, which have vertical, resp.\ horizontal,
domain and codomain, denoted by $\domv (\alpha)$, $\codv (\alpha)$,
$\domh (\alpha)$, and $\codh (\alpha)$. We picture this as, e.g.,
$\alpha$ above, where $u = \domh (\alpha)$, $u' = \codh (\alpha)$, $h
= \domv (\alpha)$, and $h' = \codv (\alpha)$. Finally, there are
operations for composing double cells: \emph{horizontal} composition
$\rond$ composes them along a common vertical morphism,
\emph{vertical} composition $\vrond$ composes along horizontal
morphisms. Both vertical compositions (of morphisms and of double
cells) may only be associative up to coherent isomorphism. The full
axiomatisation is given by Garner~\cite{GarnerPhD}, and we here only
mention the \emph{interchange law}, which says that the two ways of
parsing the above diagram coincide: $(\beta' \rond \beta) \vrond
(\alpha' \rond \alpha) = (\beta' \vrond \alpha') \rond (\beta \vrond
\alpha)$.

Returning to our playground for $\pi$, we put
\begin{definition}
  Let $\HH \subseteq \Chatf$ be the identity-on-objects subcategory of
  natural transformations with injective components, except perhaps on
  channels.
\end{definition}

\begin{wrapfigure}[4]{r}{0pt} %
  \begin{minipage}[c]{0.25\linewidth} %
    \vspace*{-1.6em} %
    \begin{equation}%
      \label{eq:cell}%
      \diag(.5,.7){%
      X \& X' \\
      U \& V \\
      Y \& Y' %
    }{%
      (m-1-1) edge[labeld={l}] (m-1-2) %
      (m-2-1) edge[labelo={k}] (m-2-2) %
      (m-3-1) edge[labelu={h}] (m-3-2) %
      (m-1-1) edge[labell={s}] (m-2-1) %
      (m-1-2) edge[labelr={s'}] (m-2-2) %
      (m-3-1) edge[labell={t}] (m-2-1) 
      (m-3-2) edge[labelr={t'}] (m-2-2) %
    } %
    \end{equation} %
  \end{minipage} %
\end{wrapfigure}
For $\Dh$, we take the full subcategory of $\HH$ spanning
  positions\footnote{Injective transformations suffice for CCS, but
    not for $\pi$,  because of channel mobility.}.
Finally, as double cells, $\D$ has commuting diagrams as on the right,
where the vertical cospans are plays and $h,k$, and $l$ are in $\HH$.

\begin{proposition}
   $\D$ forms a pseudo double category.
\end{proposition}
There is more data to provide and axioms to check to obtain that $\D$
forms a playground. The most serious challenge is to show that the
vertical codomain functor $\codv \colon \DH \to \Dh$ is a
(Grothendieck) fibration~\cite{Jacobs}.  Here, $\DH$ denotes the
category with vertical morphisms as objects, and double cells as
morphisms. The functor $\codv$ maps any play to its initial position
and any double cell to its lower border. Intuitively, the fibration
axiom amounts to the existence, for all plays $Y \xproto{u} X$ and
horizontal morphisms $X' \xto{h} X$, of a universal ($\approx$
maximal) way of restricting $u$ to $X'$, as on the left below:
\begin{mathpar}
    \diag(1,1){%
      |(X)| {Y'} \& |(Y)| {Y} \\ %
      |(U)| {X'} \& |(V)| {X} %
    }{%
      (X) edge[dashed,labela={h'}] (Y) %
      edge[pro,dashed,twol={u'}] (U) %
      (Y) edge[pro,twor={u}] (V) %
      (U) edge[labela={h}] (V) %
      (l) edge[dashed,cell=.3,labela={\alpha}] (r) %
    }
\and     \Diag(.25,1){%
    }{%
      |(U'')| E'' \\ \\ 
      \& |(U')| E' \& \& |(U)| E \\ 
      |(Y'')| p(E'') \\ \\ 
      \& |(Y')| p(E') \& \& |(Y)| p(E)  %
    }{%
      (U') edge[labelb={r}] (U) %
      (Y') edge[labelb={p(r)}] (Y) %
      (U) edge[serif cm-to,fore,shorten <=.3cm,shorten >=.3cm] (Y) %
      (U'') edge[bend left=10,labelar={t}] (U) %
      (Y'') edge[bend left=10,labelar={p(t)}] (Y) %
      (U'') edge[serif cm-to,fore,shorten <=.3cm,shorten >=.3cm] (Y'') 
      (U'') edge[dashed,labelbl={s}] (U') %
      (Y'') edge[labelbl={k}] (Y') %
      (U') edge[serif cm-to,fore,shorten <=.3cm,shorten >=.3cm] (Y') 
    } %
\end{mathpar}
Formally, consider any functor $p \colon \E \to \B$. A morphism $r
\colon E' \to E$ in $\E$ is \emph{cartesian} when, as on the right
above, for all $t \colon E'' \to E$ and $k \colon p(E'') \to p(E')$,
if $p(r) \rond k = p(t)$ then there exists a unique $s \colon E'' \to
E'$ such that $p(s) = k$ and $r \rond s = t$.
\begin{definition} 
  A functor $p \colon \E \to \B$ is a \emph{fibration} iff for all $E 
  \in \E$, any $h \colon B' \to p(E)$ has a cartesian lifting, i.e., a 
  cartesian antecedent by $p$. 
\end{definition} 
Unlike in the CCS case, what the lifting $u'$ should be in our case is
generally not obvious (see Ex.~\ref{ex:restr} and~\ref{ex:restr2} below).


\subsection{Factorisations and fibrations}
Our approach to ensuring that $\codv$ is a fibration works for $\pi$
as well as for CCS, and is much clearer conceptually than our first
proposal~\cite{DBLP:conf/calco/Hirschowitz13}.
\begin{definition}
  A (strong) \emph{factorisation system}~\cite{FK,Joyal:ncatlab:facto}
  on a category $\C$ consists of two subcategories $\LL$ and $\RR$ of
  $\C$, both containing all isomorphisms, such that any morphism in
  $\C$ factors essentially uniquely as $r \rond l$ with $l \in \LL$
  and $r \in \RR$.
\end{definition}
`Essentially unique' here means unique up to unique commuting
isomorphism. 
\begin{example}
  In $\Set$,  surjective and
  injective maps form a factorisation system. 
\end{example}
Our aim is to construct such a factorisation system $(\LL,\RR)$ on
$\Chatf$, such that $\LL$ contains all t-legs of plays, and $\RR$
contains all morphisms in $\HH$.  
The idea is to compute the restriction of $V$ along $h$, as
in~\eqref{eq:cell}, by factoring $t' \rond h$ as $k
\rond t$ with $k \in \RR$ and $t \in \LL$, and then taking the
pullback of $k$ and $s'$. 

Actually, it is enough to demand that $\LL$ contains t-legs of
\emph{seeds}.  Indeed, as is well-known, $\LL$ is always stable under
pushout; and, by construction, t-legs of plays are composites of
pushouts of t-legs of seeds.

\begin{wrapfigure}{r}{0pt}
  \begin{minipage}[t]{0.23\linewidth}
  \vspace*{-.6em}
    \diag{%
    A \& C \\
    B \& D. %
  }{%
    (m-1-1) edge[labelu={u}] (m-1-2) %
    edge[labell={f}] (m-2-1) %
    (m-2-1) edge[labeld={v}] (m-2-2) %
    (m-1-2) edge[labelr={g}] (m-2-2) %
    (m-2-1) edge[dashed,labelal={h}] (m-1-2) %
  }
\end{minipage}
\end{wrapfigure}
We rely on Bousfield's
construction~\cite{Bousfield,Joyal:ncatlab:facto} of factorisation
systems from a generating class of maps in $\LL$ (the \emph{generating
  cofibrations}).  For any morphism $f \colon A \to B$ and $g \colon
C \to D$, let $f \ortho g$ iff for all commuting squares as on the right,
there is a unique lifting $h$ making both triangles commute.  This
extends in the obvious way to classes of morphisms, which we denote by
$\LL \ortho \RR$. For all classes $\LL$ and $\RR$ of morphisms, let
$\botright{\LL} = \ens{g \aalt \LL \ortho \ens{g}}$ and $\botleft{\RR} =
\ens{f \aalt \ens{f} \ortho \RR}$.
\begin{theorem}[Bousfield]
  For any class $\TT$ of morphisms in any locally presentable category
  $\E$, the pair $(\botrightleft{\TT},\botright{\TT})$ forms a
  factorisation system.
\end{theorem}
In the setting of the theorem, one may construct the double category $\D_{\LL,\RR}$,
with $\LL = \botrightleft{\TT}$ and $\RR = \botright{\TT}$, with the same objects 
as  $\E$, and such that
\begin{itemize}
\item vertical morphisms $X \to Y$ are cospans $X \xto{f} U \xot{l} Y$ with $l \in \LL$,
\item horizontal morphisms are morphisms in $\RR$, and
\item double cells are diagrams like~\eqref{eq:cell} with $r' \in \RR$.
\end{itemize}
The theorem yields:
\begin{proposition}
  The functor $\codv \colon (\D_{\LL,\RR})_H \to (\D_{\LL,\RR})_h$ is a fibration.
\end{proposition}

{\it Proof.\ }
  Consider any vertical morphism  $X \xto{f} U \xot{l} Y$ and $r \colon Y' \to Y$ in $\RR$.
  \begin{wrapfigure}[6]{r}{0pt} %
  \begin{minipage}[c]{0.18\linewidth} %
    \vspace*{-.3em} %
    \Diag(.5,.7){%
      \pbk{m-2-1}{m-1-1}{m-1-2} %
    }{%
      X' \& X \\
      U' \& U \\
      Y' \& Y %
    }{%
      (m-1-1) edge[dashed,labeld={r''}] (m-1-2) %
      (m-2-1) edge[dashed,labelo={r'}] (m-2-2) %
      (m-3-1) edge[labelu={r}] (m-3-2) %
      (m-1-1) edge[dashed,labell={f'}] (m-2-1) %
      (m-1-2) edge[labelr={f}] (m-2-2) %
      (m-3-1) edge[dashed,labell={l'}] (m-2-1) 
      (m-3-2) edge[labelr={l}] (m-2-2) %
    } %
  \end{minipage} %
\end{wrapfigure}
We construct the restriction of $(f,l)$ along $r$ by factoring $l
\rond r$ as $r' \rond l'$, and then taking the pullback, as on the
right.  It is well-known that in any factorisation system, $\RR$ is
stable under pullback, hence $r'' \in \RR$. The universal property of
this restriction follows from the other well-known general fact that
for all $l \in \LL$ and $r \in \RR$, we have $l \ortho r$. Indeed,
consider as in Fig.~\ref{fig:cartesian} any other vertical morphism
$X'' \xto{f''} U'' \xot{l''} Y''$ and morphism $(t,t',t'')$ to $U$,
together with a morphism $s \colon Y'' \to Y'$ such that $r \rond s =
t$.  The lifting property $l'' \ortho r'$ and the universal property
of pullback give the unique $s'$ and $s''$ making the diagram of
Fig.~\ref{fig:cartesian} commute. Finally, $s'$ and $s''$ are in
$\RR$, by the general fact that $\RR$ has the \emph{left cancellation}
property~\cite{Joyal:ncatlab:facto}: for all composable $u$ and $v$,
$vu \in \RR$ and $v \in \RR$ implies $u \in \RR$.  \qed
\begin{figure}[t]
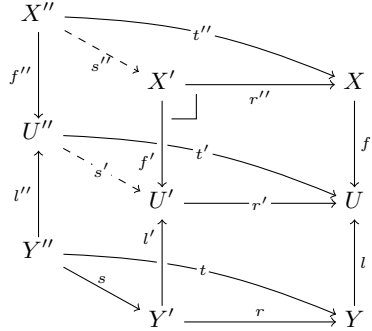

  \centering
  \Diag(.25,1){%
    \pbk{U'}{X'}{X} %
  }{%
    |(X'')| X'' \\ \\
    \& |(X')| X' \& \& |(X)| X \\
    |(U'')| U'' \\ \\
    \& |(U')| U' \& \& |(U)| U \\
    |(Y'')| Y'' \\ \\
    \& |(Y')| Y' \& \& |(Y)| Y %
  }{%
    (X') edge[labeld={r''}] (X) %
    (U') edge[labelo={r'}] (U) %
    (Y') edge[labelu={r}] (Y) %
    (X) edge[labelr={f}] (U) %
    (Y) edge[labelr={l}] (U) %
    (X'') edge[bend left=10,labelo={t''}] (X) %
    (U'') edge[bend left=10,labelo={t'}] (U) %
    (Y'') edge[bend left=10,labelo={t}] (Y) %
    (X'') edge[labell={f''}] (U'') %
    (Y'') edge[labell={l''}] (U'') 
    (X'') edge[dashed,labeld={s''}] (X') %
    (U'') edge[dashed,labelo={s'}] (U') %
    (Y'') edge[labelu={s}] (Y') %
    (X') edge[fore,labellat={f'}{.7}] (U') %
    (Y') edge[fore,labellat={l'}{.8}] (U') 
  } %
  \caption{Cartesianness of lifting}
\label{fig:cartesian}
\end{figure}

\begin{corollary}
  The functor $\codv \colon \DH \to \Dh$ is a fibration.
\end{corollary}
\begin{proof}
  We apply the theorem for $\TT$ the set of all t-legs of seeds to
  obtain a factorisation system $(\LL,\RR)$ on $\Chat$.  We
  temporarily work in $\Chat$, and then show that the needed
  factorisations remain in $\Chatf$.  Finally, we verify that
  $\HH$ is contained in $\RR$, is stable under pullback, and has the
  left cancellation property.\qed
\end{proof}
Let us briefly explain what factorisations do. First, for all
morphisms $r \colon U \to V$ in $\Chat$, $r \in \RR$ iff $t \ortho r$
for all t-legs $t$ of seeds. But for any seed $X \xto{s} M \xot{t}
Y$, $M$ is a representable presheaf. By Yoneda, $t \ortho r$ intuitively means
that for all $x \in V(M)$, if $Y$ is already present in $U$, then the
whole of $M$ is.  Otherwise said, morphisms in $r$ may not `grow' new
moves from initial positions.  Consequently, in~\eqref{eq:cell}, the
first factor $t$ of any $t' \rond h \colon Y \to V$ will perform all
possible moves from $Y$; and the second factor will then map the
result to $V$.
\begin{example}\label{ex:restr}
  Consider the synchronisation of Fig.~\ref{fig:tau}. Let $V$ be the
  synchronisation obtained by identifying $\alpha$ and $\beta$, so that the
  received name is already known to $y$. Consider the restriction of
  $V$ to just $y$. By Ex.~\ref{ex:interface}, an input move (in the
  absence of a corresponding output move) cannot receive an already
  known channel. Thus, the restriction of $V$ to $y$ is just the input
  seed.
\end{example}
\begin{example}\label{ex:restr2}
  Consider again the synchronisation of Fig.~\ref{fig:tau}, say $X
  \xto{f} U \xot{l} Y$, and the position $[3]+[1]$ consisting of a
  ternary player $x_0$ and a unary player $y_0$, not sharing any
  channel.  Consider the horizontal map $r \colon [3] + [1] \to Y$
  defined by $x_0 \mapsto x$ and $y_0 \mapsto y$.  The factorisation
  of $l \rond r$ as $r' \rond l'$ yields a play where $x_0$ does an
  $o_{3,2,3}$ move and $y_0$ does an $\iota_{1,1}$ move (the order is
  irrelevant). Thus, restrictions of moves may be plays of length
  strictly greater than one.
\end{example}
We at last obtain:
\begin{theorem}
  $\D$ forms a playground.
\end{theorem}

\subsection{Innocent Strategies}\label{sec:strats}
Following our previous work~\cite{DBLP:conf/calco/Hirschowitz13}, we
now associate to each object, i.e., position 
$X$, its category of
\emph{strategies} $\SS_X$. 
Consider first the most naive notion of strategy 
\begin{wrapfigure}[5]{r}{0pt}
  \begin{minipage}[t]{0.14\linewidth}
  \vspace*{-.6em}
    \Diag(.4,.4){%
    }{%
     |(Z)| Y' \& |(Y')| Y' \\
     |(Y)| Y \&  \\
     |(X)| X \& |(X')| X %
    }{%
      (Z) edge[pro,labell={w}] (Y) %
      (Y) edge[pro,labell={u}] (X) %
      (Y') edge[pro,twor={u'}] (X') %
      (X) edge[identity] (X') %
      (Z) edge[identity] (Y') %
      (Y) edge[cell={0.1},labela={\iso}] (r) %
    }
\end{minipage}
\end{wrapfigure}
over $X$.  
In a
playground, the analogue of the poset of plays with prefix order is
given by the category $\P(X)$ with plays $u \colon Y \proto X$ as
objects and double cells as on the right as morphisms $u \to u'$.  Let
the category $\Beh{X}$ of \emph{behaviours} on $X$ be
$\FPsh{\P(X)}$. Behaviours do not yield a satisfactory notion of
strategy:
\begin{example}
  Consider the position $X$ consisting of three players $x,y,z$
  sharing a channel $a$. Let $u_{x,y}$ denote the play where $x$ sends
  $a$ on $a$, which is received by $y$; let similarly $u_{x,z}$ denote
  the play where $x$ sends $a$ on $a$, which is received by $z$. Let,
  finally, $i_z$ denote the play where $z$ inputs on $a$. One may
  define a behaviour $B$ mapping $u_{x,y}$ and $i_z$ to a singleton,
  and $u_{x,z}$ to the empty set. Because $u_{x,y}$ is accepted, $x$
  accepts to send $a$ on $a$. Because $i_z$ is accepted, $z$ accepts
  to receive on $a$. The problem is that $B$ rejecting $u_{x,z}$
  amounts to $x$ or $z$ choosing their partners for synchronising,
  e.g., $x$ accepts to send $a$ on $a$ \emph{only to players other
    than~$z$}.
\end{example}

\begin{wrapfigure}[6]{r}{0pt}
  \begin{minipage}[t]{0.2\linewidth}
    \vspace*{-.5em}
    \Diag(.3,.1){%
      \node[inner sep=2pt] (XX) at ($(m-3-2.south) - (0,.2)$) {$X$} ;
      \path[->,draw] %
      (X) edge[labelbl={x}] (XX) %
      (X') edge[labelbr={x}] (XX) %
      ; %
    }{%
     |(Z)| [m'] \& \& |(Y')| [m'] \\ 
     |(Y)| [m] \& \&  \\ 
     |(X)| [n] \& {\ } \& |(X')| [n]
    }{%
      (Z) edge[pro,labell={w}] (Y) %
      (Y) edge[pro,labell={v}] (X) %
      (Y') edge[pro,twor={v'}] (X') %
      (X) edge[identity] (X') %
      (Z) edge[identity] (Y') %
      (Y) edge[cell={0.1},labela={\iso},labelb={\alpha}] (r) %
    }
\end{minipage}
\end{wrapfigure}
We want to rule out this kind of behaviour from our model, and our
solution is innocence.  Let \emph{basic} seeds be all seeds of the
shape $\inna$, $\outnab$, $\nun$, $\tickn$, $\forkln$, or $\forkrn$,
for $a,b \in n$. Intuitively, basic seeds follow exactly one player.
Let \emph{views} be composites of basic seeds in $\Dv$.  We now
replace our base $\P(X)$ with $\V_X$, whose objects are pairs $(v,x)$
of a view $v \colon [m] \proto [n]$ and a horizontal morphism $x
\colon [n] \to X$, i.e., by Yoneda, of a player of $X$ and a view from it.
Morphisms $v \to v'$ are cells $\alpha$ as on the right.
\begin{definition}
  Let the category $\SS_X$ of \emph{strategies} over $X$ be $\OPsh{\V_X}$.
\end{definition}
\begin{remark}
  We restrict to presheaves of finite ordinals (as opposed to
  finite sets). There is an essentially surjective embedding
  $\OPsh{\V_X} \into \FPsh{\V_X}$, so we do not really lose any
  strategy in the process, only some completeness properties. On the
  other hand, we gain the syntactic characterisation used in the next
  section.
\end{remark}

\begin{wrapfigure}[6]{r}{0pt}
  \begin{minipage}[t]{0.2\linewidth}
    \vspace*{-.5em}
    \Diag(.3,.3){%
      \node[inner sep=2pt] (XX) at ($(m-3-2.south) - (0,.3)$) {$X$} ;
      \path[->,draw] %
      (X) edge[labelbl={h}] (XX) %
      (X') edge[labelbr={h'}] (XX) %
      ; %
    }{%
     |(Z)| T \& \& |(Y')| Z' \\ 
     |(Y)| Z \& \&  \\ 
     |(X)| Y \& {\ } \& |(X')| Y'
    }{%
      (Z) edge[pro,labell={w}] (Y) %
      (Y) edge[pro,labell={u}] (X) %
      (Y') edge[pro,twor={u'}] (X') %
      (X) edge[labela={r}] (X') %
      (Z) edge[labela={s}] (Y') %
      (Y) edge[cell={0.2},labela={\alpha}] (r) %
    }
\end{minipage}
\end{wrapfigure}
To relate strategies and behaviours, consider the category $\P_X$ with
as objects pairs $(u,h)$ of a play $u \colon Z \proto Y$ and a
horizontal morphism $h \colon Y \to X$, and as morphisms $(u,h) \to
(u',h')$ all diagrams as on the right.  This category contains both
$\V_X$ and $\P(X)$ in obvious ways, and it furthermore allows to
describe the views $(v,x)$ of a general play $u'$ by taking $h' =
\id_X$%
\footnote{There is a small problem, however: morphisms should only
  describe how $u$ maps to $u'$, not $w$. We actually consider a
  quotient of morphisms to rectify this.}%
. So our morphisms account both for prefix inclusion, and for
`spatial' inclusion, i.e., inclusion of a play into a play on a bigger
position.

Right Kan extension and restriction along $\op{\V_X} \into
\op{\P_X} \otni \op{\P(X)}$ induce a functor $\SS_X \to
\Beh{X}$. Intuitively, this functor maps any strategy $S$ to the
behaviour accepting a play $u$ iff $S$ accepts all views of $u$. This
also allows to view strategies as sheaves~\cite{MM} for a certain
Grothendieck topology on $\P_X$, as explained in previous
papers~\cite{2011arXiv1109.4356H}. Intuitively, strategies provide
(locally determined) behaviours for all subpositions of $X$. Such
local `behaviour' may become irrelevant when passing to the
globally-defined behaviours. In particular, $\SS_X \to \Beh{X}$ is
neither injective on objects, nor full, nor faithful.
\begin{example}
  If two strategies differ, but are both empty on the views of some
  player, then both are mapped to the empty behaviour.
\end{example}

\section{Bridging the gap with $\pi$}\label{sec:pi}
\subsection{Syntax and transition system for strategies}
One of the main results about
playgrounds~\cite{DBLP:conf/calco/Hirschowitz13} entails that
strategies over $\D$ are entirely described by the following typing rules
\begin{mathpar}
\inferrule{\ldots \ n_B \vdash S_B \ \ldots \ {(\forall B \colon [n_B] \to [n])} 
}{
n \vdashdefinite \langle (S_B)_{B \in \MMMB_n} \rangle
}
\and
\inferrule{\ldots \  n \vdashdefinite D_i \  \ldots \ (\forall i \in m)}{
n \vdash \oplus_{i \in m} D_i}~(m \in \Nat),
\end{mathpar}%
where $\MMMB_n$ is the set of basic seeds from $[n]$ as defined
above. The rules feature two kinds of judgements, $\vdash$ for plain
strategies, and $\vdashdefinite$ for \emph{definite} strategies,
intuitively those with exactly one initial state.
\begin{remark}
  The sum $\oplus$ is not commutative (although it is up to fair
  testing equivalence).
\end{remark}
\begin{theorem}[\cite{DBLP:conf/calco/Hirschowitz13}]\label{thm:syntaxstrat}
  Strategies over $[n]$ are in bijection with possibly infinite terms
  in context $n$.

  Furthermore, giving a strategy over any position $X$ amounts
  to giving a strategy over $[n]$ for each $n$-ary player of $X$.
\end{theorem}

Theorem~\ref{thm:syntaxstrat} yields the following coinductive
interpretation of processes:
\begin{center}
  ${\transl{\Gam \vdash \sum_i \alpha_i.P_i} = \with{B \mapsto
    \oplus_{\ens{i \aalt \transl{\alpha_i} = B}} \transl{\Gam \cdot \alpha_i \vdash P_i}} }$
  \\
  ${\transl{\Gam \vdash P \para Q} = \left \langle
    {\begin{array}[c]{rcl}
        \forklof{\Gam} & \mapsto & \transl{\Gam \vdash P} \\
        \forkrof{\Gam} & \mapsto & \transl{\Gam \vdash Q} \\
        - & \mapsto & \emptyset
      \end{array}} \right \rangle }$
\hfil
  ${\transl{\Gam \vdash \nu.P} = \left \langle
    {\begin{array}[c]{rcl}
        \nuof{\Gam} & \mapsto & \transl{\Gam+1 \vdash P} \\
        - & \mapsto & \emptyset
      \end{array}} \right \rangle}$,
\end{center}%

\noindent with $\transl{\send{a}{b}} = o_{\Gam,a,b}$, $\transl{a} =
\iota_{\Gam,a}$,  $\transl{\tick} = \tickof{\Gam}$,  %
 $\emptyset$ is the empty $\oplus$ sum, and  $- \mapsto \emptyset$ 
means that all unmentioned basic seeds are mapped to $\emptyset$.
\begin{example}\label{ex:transl}
  Omitting typing contexts, we have 
  $$\transl{\Gam \vdash a.P + a.Q + \send{b}{c}.R} = 
  \left \langle
          {\begin{array}[c]{rcl}
              \iota_{\Gam,a} & \mapsto & \transl{\Gam + 1 \vdash P} \oplus \transl{\Gam + 1 \vdash Q} \\
              o_{\Gam,b,c} & \mapsto & \transl{\Gam \vdash R} \\
            - & \mapsto & \emptyset
          \end{array}} \right \rangle.$$  
\end{example}
We now define a transition system for definite strategies, which is
useful for characterising fair testing equivalence, and for which we
need to define two auxiliary operations.  The first is an operation of
\emph{derivation} along a basic seed, defined from definite strategies
to strategies by $\deriv_B \langle (S_{B'})_{B' \in \MMMB_n} \rangle =
S_B$.  The second is a partial \emph{restriction} operation from
strategies to definite strategies, defined if $i \in p$ by
$\restr{(\oplus_{i' \in p} D_{i'})}{i} = D_i$.
\begin{example}
  Following up on Example~\ref{ex:transl} and omitting contexts, we have
  \begin{mathpar}
    \restr{(\deriv_{\transl{a}} \transl{a.P + a.Q +
        \send{b}{c}.R})}{2} = \transl{Q} \and
    \restr{(\deriv_{\transl{\send{b}{c}}} \transl{a.P + a.Q +
        \send{b}{c}.R})}{1} = \transl{R}.
  \end{mathpar}
\end{example}
These operations may be extended to arbitrary strategies and moves, in
a way which we will gloss over here. We may thus write
$\restr{(\deriv_M S)}{i}$. This yields:
\begin{definition}
  Let $\SSS_\D$ denote the free reflexive graph with as vertices pairs
  of a position $X$ and a definite strategy $D$ over $X$, and as edges
  all well-defined triples $(X,D) \xto{M} (Y, \restr{(\deriv_M
    D)}{i})$, for all moves $M \colon Y \proto X$.
\end{definition}
We view this graph as a transition system for strategies.
\begin{example}
  We have examples mirroring transitions in $\pi$. Calling $D$ the
  translation of the process of Example~\ref{ex:transl}, we have,
  e.g., $([\Gam], D) \xto{\iota_{\Gam,a}} ([\Gam+1], \transl{P})$, and the same
  with $Q$. But we also have transitions for things which usually go
  into structural equivalence, e.g., 
  $([\Gam],\transl{P \para Q}) \xto{\forkof{\Gam}} ([\Gam]\para[\Gam], (\transl{P},\transl{Q})).$
  In the final state, by the second part of
  Theorem~\ref{thm:syntaxstrat}, we define a strategy on
  $[\Gam]\para[\Gam]$ by providing two strategies on $[\Gam]$.
  Similarly, we have a transition 
  $([\Gam],\transl{\nu.P}) \xto{\nuof{\Gam}} ([\Gam+1],\transl{P})$.
\end{example}

\subsection{Fair testing equivalence from the transition system}
The point of the transition system $\SSS_\D$ is to characterise our
semantic analogue of fair testing equivalence. For lack of space, we
describe the characterisation, omitting the direct, game semantical
definition.  First, as announced in the introduction, we allow tests
to rename some channels.  Recall from Definition~\ref{def:interface}
the canonical interface $I_X$ of a position $X$.
\begin{definition}
  For any state $(X,D)$ of $\SSS_\D$, a \emph{test} for $(X,D)$ is a pair
  of a horizontal morphism $h \colon I_X \to Y$ and a strategy $T$ on
  $Y$. 
\end{definition}
The morphism $I_X \to Y$ may identify some channels and introduce new
ones.  Whether such a test is passed successfully will be determined
by the `closed-world' dynamics of the strategy $(D, T)$ over the
pushout $Z = X +_{I_X} Y$. Intuitively, the players of $Z$ are
partitioned into players from $X$ and players from $Y$, so $(D,T)$ is,
by a slight abuse of language, a strategy for the whole.

Let now $\SSS_\D^\W$, the \emph{closed-world} part of $\SSS_\D$, be the
identity-on-vertices subgraph of $\SSS_\D$ consisting of edges whose
underlying moves have the shape $\taunamcd$, $\nun$, $\tickn$, or
$\forkn$. There is an obvious morphism of reflexive graphs $\labelD
\colon \SSS_\D^\W \to \Sierp$ to the one-vertex reflexive graph with
one non-identity edge $\tick$. We denote by $(X,D) \xTo{} (X',D')$ the
existence of a path in $\SSS_\D^\W$ mapped by $\labelD$ to a path of
identities in $\Sierp$, and by $(X,D) \xTo{\tick} (X',D')$ the
existence of a path mapped to a path consisting of identities and
exactly one $\tick$ edge.
\begin{definition}
  Let $\bot^{\D}$ denote the set of all vertices $x$ of $\SSS_\D$ such that,
  for all $x \xTo{} x'$, there exists $x''$ such that $x' \xTo{\tick}
  x''$.  Let $(X,D)^\bot$ denote the set of all tests $(h,T)$ such
  that $(D,T) \in \bot^{\D}$. Finally, let $(X,D) \faireqof{\D}{} (X',D')$ iff
  $(X,D)^\bot = (X',D')^\bot$.
\end{definition}

\subsection{Main results}
In this section, we at last state our main results.  First, let us
define our variant of fair testing equivalence for $\pi$.  Let a
\emph{test} for $\Gam \vdash P$ consist of a pair of a map $h \colon
\Gam \to \Del$ and a process $\Del \vdash R$.  Let $\picalc^\W$ denote
the identity-on-vertices sub-reflexive graph of $\picalc$ consisting
of $\tau$ and $\tick$ transitions.  There is an obvious morphism $\labelpi
\colon \picalc^\W \to \Sierp$ and, mimicking previous notation, we put:
\begin{definition}
  Let $\bot^\picalc$ denote the set of all vertices $x$ of $\picalc$
  such that, for all $x \xTo{} x'$, there exists $x''$ such that $x'
  \xTo{\tick} x''$.  Let $(\Gam \vdash P)^\bot$ denote the set of all
  tests $(h,R)$ such that $(P[h] \para R) \in \bot^\picalc$. Finally,
  let $(\Gam \vdash P) \faireqof{\picalc}{} (\Gam \vdash Q)$ iff $(\Gam \vdash
  P)^\bot = (\Gam \vdash Q)^\bot$.
\end{definition}
\begin{theorem}\label{thm:1}
  For all $P,Q$,  $(\Gam \vdash P)
  \faireqof{\picalc}{} (\Gam \vdash Q)$ iff $([\Gam], \transl{P})
  \faireqof{\D}{} ([\Gam], \transl{Q})$.
\end{theorem}
\paragraph{Proof sketch.} 
The main difficulty is that we have to compare \ltss{} over very
different alphabets. A first point is that edges in $\SSS_\D$ are very
intensional. E.g., an input transition describes not only the involved
channels but also which \emph{player} makes the move.  A second point
is that $\SSS_\D$ is not `modular', in the sense that it is not
obvious to infer the transitions of a vertex $(X,D)$ from the
transitions of players of $X$.  E.g., we have transitions $\transl{\nu
  a.a(x)} \xto{\nu_0} \transl{a(x)} \xto{\iota_{1,1}} 0$.  

We rectify the latter deficiency first, by designing a finer \lts{}
$\SSS^\LLL_\D$ for $\D$.  Its vertices are triples $(I,h,S)$ of an
interface $I$, a horizontal map $h \colon I \to X$, and a strategy $S$
over $X$.  $I$ represents all channels known to the environment, and
the idea is that all transitions in $\SSS^\LLL_\D$ may be completed
into closed-world transitions by interacting at $I$. This corrects
the second mentioned deficiency, but $\SSS^\LLL_\D$ remains too
intensional. %
\begin{figure}[t]  
  \centering  
  \begin{mathpar} 
  \inferrule{ }{%
    (\Del \xto{h} \Gam)  %
    \xot{\tick} %
    (\Del \xto{h} \Gam)  %
  } 
  \and 
  \inferrule{ }{ 
    (\Del \xto{h} \Gam)   
    \xot{\nu}  
    (\Del \xto{h} \Gam \xinto{\subseteq} \Gam+1) %
  } %
  \and 
  \inferrule{a \in \im(h) }{ 
    (\Del \xto{h} \Gam)  \xot{\iota^{}(a)} %
    (\Del + 1 \xto{h+!} \Gam+1)  
  } 
  \and 
  \inferrule{ a \in \im(h)}{ 
    (\Del \xto{h} \Gam)  \xot{o^{}(a,b)}  
    (\Del + 1 \xto{[h,b]} \Gam) %
  } 
  \and 
  \inferrule{a,c \in \im(h); a \neq c }{%
    (\Del \xto{h} \Gam)  %
    \xot{%
      o(a,b) \paradr \iota(c)} %
    (\Del \xto{h} \Gam \xto{\subseteq} \Gam + 1) %
}%
\and 
  \inferrule{ 
  }{%
    (\Del \xto{h} \Gam) \xot{\delta}  
    (\Del \xto{h} \Gam) 
  }~\cdot 
\end{mathpar} 
  \caption{Edges for $\A$}  
  \label{fig:A}  
\end{figure} %
So, we coarsen the \lts{} $\SSS^\LLL_\D$ to a new \lts{} $\SSS^\A_\D$.
The new labels are given by the free reflexive graph $\A$
with as vertices all maps $\Del \to \Gam$ of finite sets, 
and as with edges as defined by the rules in Fig.~\ref{fig:A}\footnote{In Fig.~\ref{fig:A}, we  
  put side conditions as premises for conciseness.}.    

The idea, for vertices, is that $\Del$ represents the channels of the
interface, and $\Gam$ represents the channels that the considered
process or strategy (say, an agent) knows locally.  The first rule
should be easy. The second rule says that an agent may create a
private channel, \emph{a priori} unknown to the environment. The next
two rules, for input and output, have been simplified for clarity. The
important point is their symmetry: both add one channel to the
interface. The input rule, however, locally considers the new channel
as fresh, whereas the output rule records that it is the sent
channel. By the side condition, the channel on which the
synchronisation occurs should belong to $\Del$.  The rules for input
and output describe one way of decomposing a synchronisation.  The
last two rules describe another way, where an input on $a$ and an
output on $c$ both occur for the same agent, which cannot verify
locally that $a = c$.  Again, we only present a particular case of our
real rules (actually this is just the case $b \notin \im(h)$). These
rules are reminiscent of Rathke and Soboci\'{n}ski~\cite{modularLTS}
(for input/output), and Crafa et
al.~\cite{DBLP:conf/fossacs/CrafaVY12} (for partial synchronisations).

We have already mentioned that $\SSS_\D$ may be viewed as \anlts{}
$\SSS^\A_\D$ over $\A$.  It is not too much work to also view
$\picalc^\A$ as \anlts{} over $\A$.  Next, we define when two
transitions in $\A$ are complementary, i.e., are the restrictions of a
closed-world transition. This gives the right notion of
complementarity for both $\picalc^\A$ and $\SSS^\A_\D$, so that fair
testing equivalence in $\picalc$ and $\SSS_\D$ may be checked in terms
of transitions over $\A$. Thus, in order to check whether an agent $P$
passes a test $T$, e.g., instead of considering transition sequences
$P \para T \xTo{} Q$, one may consider complementary sequences $P
\xTo{w}_\A P'$ and $T \xTo{w'}_\A T'$ such that $Q = P' \para T'$.

Thanks to this, one reduces to proving that the translation
$\translfun \colon \picalc \to \SSS_\D$ is surjective up to weak
bisimilarity (except for empty strategies), which ensures that there
are enough tests in $\picalc$. For this, the only subtlety is that in
$\picalc$, $\nu$ is a standalone construct, which may not be part of a
guarded sum, while in $\SSS_\D$ it is treated exactly as inputs,
outputs, and ticks. This is dealt with by encoding any guarded sum
$\nu.P + \ldots$ as, informally, $\nu c.(\bar{c}.\nu.P + \ldots)$.
\qed

In the course of our proof, we have shown that almost all
strategies are weakly bisimilar, hence fair testing equivalent, to
some $\transl{P}$. Actually, the only strategy which is not is
$\emptyset$, which is in fact fair testing equivalent to
$\transl{\tick}$!  This entails
\begin{theorem}\label{thm:2}
  For all strategies $S$ over $[\Gam]$, there exists a
  process $\Gam \vdash P$ such that $\transl{\Gam \vdash P}
  \faireqof{\D}{} ([\Gam],S)$.
\end{theorem}

Theorems~\ref{thm:1} and~\ref{thm:2} together are the desired full abstraction result.

\bibliography{../common/bib}
\bibliographystyle{abbrv}

\end{document}